\documentclass[runningheads]{llncs}
\usepackage[utf8]{inputenc}

\title{Accountable Safety for Rollups}

\usepackage{import}
\usepackage{bbm}
\usepackage{amsfonts}
\usepackage{amssymb}
\PassOptionsToPackage{hyphens}{url}\usepackage{hyperref}
\usepackage{mathtools}
\usepackage{listings}
\usepackage{import}
\usepackage{IEEEtrantools}
\usepackage{algorithm}
\usepackage{algpseudocode}
\usepackage{varwidth}
\usepackage{enumitem}
\usepackage{thmtools}
\usepackage{amsthm}
\usepackage{thm-restate}
\usepackage{tabularx}
\usepackage{array}
\usepackage{amsmath}

\usepackage{IEEEtrantools}
\usepackage{xspace}

\hypersetup{breaklinks=true}

\newcommand{\tx}[0]{\ensuremath{\mathsf{tx}}}
\newcommand{\txs}[0]{\ensuremath{\mathsf{txs}}}
\newcommand{\LOG}[0]{\ensuremath{\mathsf{LOG}}}
\newcommand{\ch}[0]{\ensuremath{\mathsf{ch}}}

\newcommand{\ie}[0]{\emph{i.e.}\xspace}
\newcommand{\eg}[0]{\emph{e.g.}\xspace}
\newcommand{\cf}[0]{\emph{cf.}\xspace}
\newcommand{\st}{\ensuremath{\mathsf{st}}}
\newcommand{\genesisstate}{\ensuremath{\mathsf{st}_0}}
\newcommand{\stc}{\ensuremath{\left<\mathsf{st}\right>}}

\newcommand{\client}{\ensuremath{\mathcal{V}}}
\newcommand{\transition}{\ensuremath{\delta}}

\newcommand{\Tfull}{\ensuremath{T_\mathrm{cf}}}
\newcommand{\Tlight}{\ensuremath{T_\mathrm{cl}}}
\newcommand{\Tfraud}{\ensuremath{T_\mathrm{fraud}}}
\newcommand{\Th}{\ensuremath{\mathrm{Th}}}
\newcommand{\Tconf}{\ensuremath{T_\mathrm{conf}}}

\usepackage{tabularx}
\usepackage{authblk}
\usepackage{hyperref}

\usepackage{etoolbox}
\newtoggle{fullversion}
\toggletrue{fullversion}   

\iftoggle{fullversion}
{
  \author{Ertem Nusret Tas\footnote{Part of this work was done by the author as part of an internship at Celestia.}\inst{1} \and John Adler\inst{2} \and Mustafa Al-Bassam\inst{2} \and Ismail Khoffi\inst{2} \and David Tse\inst{1} \and Nima Vaziri\inst{3}}
    \institute{Stanford University\\
    \email{\{nusret,dntse\}@stanford.edu}\\
    \and
    Celestia\\
    \email{\{john,mustafa,ismail\}@celestia.org}
    \and
    Polychain Capital\\
    \email{nima@polychain.capital}
}
}
{
  \author{}
  \institute{}
}

\begin{document}

\maketitle

\begin{abstract}
Accountability, the ability to provably identify protocol violators, gained prominence as the main economic argument for the security of proof-of-stake (PoS) protocols. 
Rollups, the most popular scaling solution for blockchains, typically use PoS protocols as their parent chain. 
We define accountability for rollups, and present an attack that shows the absence of accountability on existing designs. 
We provide an accountable rollup design and prove its security, both for the traditional `enshrined' rollups and for sovereign rollups, an emergent alternative built on lazy blockchains, tasked only with ordering and availability of the rollup data.
\end{abstract}

\section{Introduction}
\label{sec:introduction}

\begin{quote}
“Man is condemned to be free; because once thrown into the world, he is responsible for everything he does. It is up to you to give (life) a meaning.”
-- Jean-Paul Sartre
\end{quote}

Infamously, blockchains have exhibited low transaction throughput.
This has been attributed to what is colloquially referred to as the \emph{blockchain scalability trilemma}~\cite{vitalik-sharding}, whereby at most two of throughput, decentralization, and security can be achieved simultaneously by a blockchain.
Various scaling proposals have emerged over the years, including state and payment channels~\cite{state_channels}, sidechains~\cite{sidechains}, Plasma~\cite{plasma}, and sharding~\cite{divide-and-scale}.
More recently, rollups~\cite{adler2020building,zk-rollup} have garnered interest as a promising method of scaling blockchains with minimal compromises.
At a high level, a rollup is a system where transactions are posted to a separate existing blockchain called the \emph{parent chain}, while state transitions are computed by rollup nodes that are distinct from the nodes maintaining the parent chain.
The security of rollups rests on the fact that any actor can verify the rollup state transitions by referring to the data on the parent chain, and does not have to ask a trusted third party.

Rollups have typically been classified into \emph{optimistic} and \emph{ZK} varieties depending on the method used by the rollup clients to succinctly validate the state transitions.
Both optimistic and ZK rollups have traditionally enshrined their parent chain as a \emph{settlement} layer that has the ultimate authority over the validity of the rollup state.
For instance, Ethereum rollups use the Ethereum beacon chain as their settlement layer, where the new state commitments and the associated validity or fraud proofs are posted. 
The beacon chain verifies these proofs using light-weight methods to validate the posted state commitments.
This relieves the burden of checking the validity of the rollup state from the rollup clients, which can then refer to the Ethereum beacon chain to learn the correct latest state commitment.
We will call such rollups \emph{enshrined rollups}.

Recently, \emph{sovereign rollups} emerged as an alternative construction, where no blockchain is enshrined as a settlement layer, and the clients are free to choose among different methods or services for validating and serving the rollup state.
For instance, clients of a sovereign rollup can opt to validate the state commitments themselves using validity proofs provided by the rollup nodes.
In this case, nodes of the parent chain are only tasked with maintaining the ordering and availability of the rollup data, and unlike in an enshrined rollup, are not enstrusted with verifying the correctness of the state.

Security of rollups depends on the security of the parent chain protocol (\eg, Proof-of-Stake Ethereum).
This protocol is often executed by staked rational agents.
To incentivize honest behavior among these agents, Buterin and Griffith~\cite{casper-ffg} advocate the notion of \emph{accountable safety}, which enhances the traditional security guarantees with the ability to identify malicious nodes in the event of a safety violation.
Accountability comes with economic punishments on the nodes accused of adversarial behavior, thus discourages safety attacks.
This \emph{trust-minimizing} notion of security has gained prominence as the main argument supporting the cryptoeconomic security of PoS Ethereum~\cite{gasper} and Cosmos zones based on Tendermint~\cite{tendermint} (\cf~\cite{forensics} for a comprehensive analysis of other protocols with accountable safety, also called forensic support).
Despite forming the bedrock of security for many L1 blockchains, accountable safety has remained an elusive notion for both enshrined and sovereign rollups.

\paragraph{Contributions.}
The main contributions of this work are
\begin{enumerate}
    \item definition of sovereign rollups in comparison to the enshrined rollups,
    \item definition of accountability for rollups,
    \item an attack demonstrating the lack of accountability on existing rollup architectures, and
    \item design and security analysis of accountable rollups.
\end{enumerate}

\begin{figure}[h]
    \centering
    \includegraphics[width=1\linewidth]{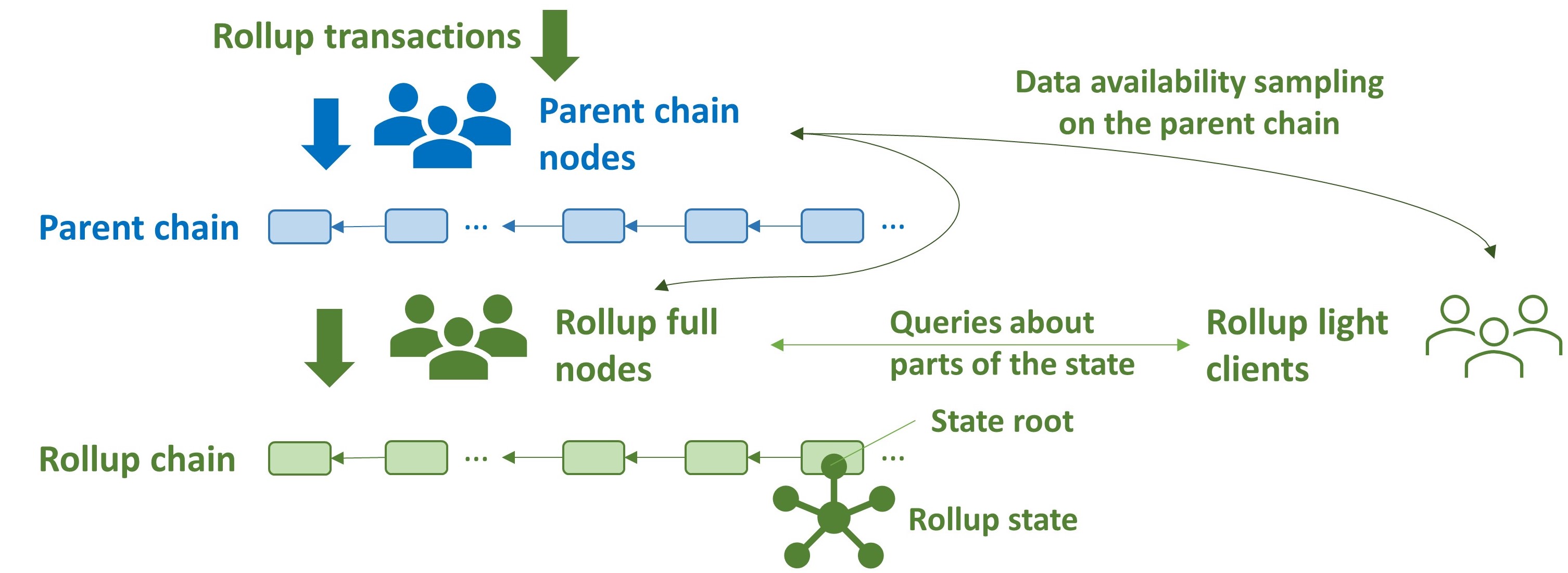}
    \caption{An accountable, sovereign rollup. Parent and the rollup chains are shown in blue and green respectively. Rollup full nodes download the rollup transactions ordered by the parent chain and construct the rollup chain. Light clients learn about the rollup state and the associated proofs from the full nodes. Both nodes do data availability sampling on the parent chain blocks for accountability.}
    \label{fig:rollup-diagram}
\end{figure}

Accountability requires that whenever the rollup clients obtain commitments to conflicting rollup states, the staked nodes of the parent chain responsible for this safety breach must be identified and slashed.
However, since the clients do not download the parent chain data that is not relevant to the rollup, they can be tricked into accepting rollup transactions within \emph{unavailable} parent chain blocks.
Once the parent chain identifies these blocks and restarts from an available block, the future rollup nodes might conflict with the old ones that downloaded their data from the unavailable, discarded blocks.
To avoid this, clients of an accountable rollup verify the data availability of the parent chain via light-weight methods such as data availability sampling~\cite{albassam2019fraud} (Figure~\ref{fig:rollup-diagram}).
They also validate the finality of the parent chain blocks by downloading their headers.
Using the rollup transactions within the finalized and available parent chain, the full nodes construct the rollup chain, whereas the light clients query them for the latest state commitment of the rollup chain. 
The honest full nodes give this commitment along with a proof of validity, which convinces the client that it is the latest valid commitment.

\paragraph{Outline.}
In Section~\ref{sec:preliminaries}, we define security and accountability for rollup full nodes and light clients.
In Section~\ref{sec:rollup-system}, we formulate a general model for all rollup types.
In Section~\ref{sec:accountable-rollup}, we show an attack on the accountable safety of the existing rollup designs, and describe how they can be made accountable.
In Section~\ref{sec:security}, we state and prove the theorems for accountable safety and liveness of rollups.
Finally, in Section~\ref{sec:analysis}, we provide an in-depth analysis of rollups in terms of their comparative security assumptions (Section~\ref{sec:comparison-security}), scalability (\ref{sec:analysis-throughput}), latency (\ref{sec:appendix-latency}), incentive-compatibility (\ref{sec:appendix-incentives}), privacy (\ref{sec:appendix-privacy}) and execution models (\ref{sec:appendix-deployment}).
We believe that they are of independent interest.

\paragraph{Related Work.}
To our knowledge, this is the first work to analyze sovereign rollups together with the enshrined varieties, and show accountability as a feature of these rollups.
Our analysis leverages insights from a line of work on systematizing blockchain scaling solutions.
Yee et al. characterized different notions of finality for ordering transaction, agreeing on the state and checkpointing blockchains against long range attacks~\cite{long-range-survey}, and identified the properties of an ideal L2 design~\cite{layer-2-finality}.
McCorry et al. proposed a general model capturing different methods for verifying the security of assets on validating bridges, and identified the core research problems for designing bridges with an extensive literature review~\cite{layer-2-bridges}.
Gudgeon et al. provided an SoK across the various L2 systems such as payment and state channels, side chains and the bisection method for succinct verification of state updates~\cite{layer-2-sok}.

Sovereign rollups and lazy blockchains on which they are built were first proposed by Mustafa Al-Bassam~\cite{albassam2019lazyledger}.
Lazy blockchains order and serve rollup data without executing, or otherwise interpreting it.
Sovereign rollups and their comparison to the enshrined varieties were further investigated in blog posts~\cite{sovereign-rollup-blog,rollup-forum}.

Accountable safety was first defined in the blockchain setting by Buterin and Griffith as a property of Casper FFG~\cite{casper-ffg}, the finality gadget of PoS Ethereum~\cite{gasper}.
Sheng et al. characterized accountable safety as a forensic method for identifying malicious nodes after a safety violation, and proved that BFT protocols such as PBFT~\cite{pbft} and HotStuff~\cite{yin2018hotstuff} support this method.
Building on this idea, Neu et al. proposed the first accountably safe finality gadget design~\cite{aa-dilemma}.

\section{Model}
\label{sec:preliminaries}

A rollup architecture consists of a \emph{parent chain}, maintained by the \emph{parent chain nodes}, and rollup nodes, which can be classified into \emph{rollup full nodes} and \emph{rollup light clients}.
The parent chain is responsible for ordering and availability of the rollup transactions, whereas the transactions are execution and the state transitions are validated by the rollup nodes.


\paragraph{Notation.} For a positive integer $n$, we denote $\{1, \cdots, n\}$ by $[n]$.
We let $\lambda$ be the security paramemter and say that an event happens with \emph{negligible probability}, if its probability is $o(1/\lambda^d)$ for all $d > 0$. 
An event happens with overwhelming probability if it happens except with probability negligible in $\lambda$.
A sequence $A$ conflicts with sequence $B$ if neither $A$ is a prefix of $B$, nor $B$ is a prefix of $A$.

\subsection{Nodes}
\label{sec:nodes-overview}
\vspace{-0.3cm}
\noindent
\paragraph{Parent chain nodes.}
These nodes receive transactions as input, and execute a state machine replication (SMR or blockchain) protocol such as PoS Ethereum~\cite{gasper}.
Their goal is to ensure that the observers of the parent chain obtain the same sequence of totally-ordered, finalized transactions called the \emph{ledger}.
These observers typically query the parent chain nodes to obtain the correct ledger.

Although the parent chain might have its own execution logic, we assume that its nodes do not execute the \emph{rollup transactions}, \ie, the chain is lazy toward the transactions.
However, they ensure that the rollup transactions are ordered within the ledgers obtained by the observers.
Although this ledger may contain transactions for multiple rollups, the total order imposed on the transactions enables nodes of each rollup to retrieve a subsequence of their own transactions from the ledger.
In the context of a single rollup, we will hereafter call this subsequence the \emph{rollup ledger}.
We will denote by $\LOG^\client_t$ the rollup ledger in the view of a rollup full node $\client$ at time $t$.

\noindent
\paragraph{Rollup full nodes.}
These nodes (full nodes for short) execute the rollup transactions to obtain and store the rollup \emph{state}.
This state is uniquely determined by the rollup ledger.
An empty ledger corresponds to a fixed \emph{genesis state}, $\genesisstate$.
The state of a non-empty ledger can be found by applying rollup transactions iteratively on top of the previous state, starting at the genesis state.
This process is captured by a transition function $\transition$. 
For instance, the state associated with the sequence $\tx_1 \cdots \tx_n$ is $\transition(\cdots \transition(\genesisstate, \tx_1), \cdots, \tx_n)$.
We use the shorthand notation $\transition^*$ to apply a sequence $\tx_1 \cdots \tx_n$ to a state: $\transition^*(\genesisstate, \overline{\tx}) = \transition(\cdots \transition(\genesisstate, \tx_1), \cdots, \tx_n)$.
Some rollup transactions might be \emph{invalid} with respect to the current state.
By convention, if transaction $\tx$ is invalid with respect to $\st$, we let $\transition(\st,\tx) = \st$.
We will denote the current rollup state in the view of a rollup full node $\client$ at time $t$ by $\st^\client_t := \transition^*(\genesisstate, \LOG^\client_t)$.

Rollup full nodes act as light clients of the parent chain as they do not want to download transactions from the parent chain that do not concern their rollup.
They download and verify the block headers on the finalized portion of the parent chain (\eg, the prefix of the longest chain in PoW Ethereum), transaction of their rollup, but not all of the transactions within the blocks.
Then, using Merkle proofs, the full nodes can verify the inclusion of rollup transactions within the parent chain, without downloading any non-rollup transactions.

Each rollup state $\st$ can be uniquely represented by a succinct and binding commitments called the \emph{state root} and denoted by $\stc$.
For instance, a state consisting of account balances can be stored as a sparse Merkle tree of key-value pairs, and represented by the Merkle root of the tree.
We denote the current state root in the view of a rollup full node or light client $\client$ at time $t$ by $\stc^\client_t$.

\noindent
\paragraph{Rollup light clients}
Goal of a rollup light client (light client for short) is to learn a particular state element, \eg, its own account balance, within the latest rollup state.
Examples of light clients include application users and the nodes on other rollups with a rollup bridge~\cite{layer-2-bridges}.

Rollup light clients are also light clients of the parent chain.
They download and verify the finality of the headers on the parent chain, thus can verify the inclusion of any succinct data by the parent chain via Merkle proofs.
However, they do not download the block contents and the rollup transactions.
A light client is called \emph{succinct} if the data downloaded and executed by the client to obtain the latest state root is at most poly-logarithmic in the number of rollup transactions and the state size (excluding the header chain).
Thus, a succinct light client cannot obtain the current state root by downloading or executing the rollup transactions.
It instead retrieves the root from the full nodes, and validates it using different techniques (\cf Section~\ref{sec:rollup-types}).

\subsection{Environment, Adversary and the Network}
\label{sec:environment-and-adversary}

Time is slotted and and the clocks of the nodes are synchronized\footnote{Bounded clock offsets can be captured by the network delay.}.
Transactions are input to the nodes by the environment $\mathcal{Z}$.
We assume a permissionless setting for the rollup nodes, where $\mathcal{Z}$ can spawn new nodes over time.
Upon joining the protocol at some slot $t$, the new node receives every transaction and message sent prior to slot $t$.
The adversary $\mathcal{A}$ is a probabilistic poly-time algorithm.
It can corrupt rollup full nodes and parent chain nodes, then called \emph{adversarial}, whereas the remaining nodes are called \emph{honest}.
(All rollup light clients are assumed to be honest as they are passive observers; an adversarial client can only hurt itself.)
Adversarial nodes can deviate from the protocol arbitrarily (Byzantine faults) under $\mathcal{A}$'s control, which can access their internal states.
The adversary also controls the delivery of messages sent by the honest nodes, and observe all messages sent at a given slot before sending its message at that slot.
However, the adversary is required to deliver all messages sent by an honest node to \emph{all} recipients within $\Delta$ slots, where $\Delta$ is known (synchronous network).

\subsection{Data Availability}
\label{sec:data-availability}

We say that a parent chain block is unavailable if an observer sees its header before slot $t-\Delta$, yet does not see the block data by slot $r$.
Otherwise, it is said to be available.
Parent chain nodes download their blocks to ascertain their availability.
However, downloading the whole block is infeasible for succinct light clients, and undesirable for full nodes that only want the rollup data.
They can instead use \emph{data availability sampling (DAS)}, to efficiently verify their availability.

To facilitate DAS, the parent chain nodes divide their blocks into small chunks. 
They then encode these chunks with forward error-correcting codes (\eg Reed-Solomon codes~\cite{albassam2019fraud}, LDPC codes~\cite{codedMerkleTrees}), and commit to them by polynomial commitments (\eg KZG commitments~\cite{kate,dankrad-kate}) or Merkle trees.
During DAS, a rollup node queries the parent chain nodes to reveal a constant number of the committed chunks.
If it observes and verifies all of the requested chunks against the commitment, then the rollup node deems that the data is available.
If there are sufficiently many queries for distinct chunks, DAS guarantees the recovery of the transaction data.
Moreover, assuming an enhanced network model (\eg, a mixnet~\cite{mixnet}) with sufficiently many DAS queries that cannot be linked to the senders, DAS ensures that light clients will receive back chunks and conclude the availability of data if and only if it is recoverable from the revealed chunks.

\begin{definition}[Security for DAS~\cite{albassam2019fraud}]
A DAS scheme is said to be secure if it satisfies:
\begin{itemize}
\item \textbf{Soundness.} If an honest node (of any kind) accepts a parent chain block as available at time $t$, then at least one honest parent chain node has the full block data by time $t+\Delta$.  
\item \textbf{Agreement} If an honest node accepts a parent chain block as available at time $t$, then all other honest nodes accept it as available by time $t+\Delta$.
\end{itemize}
\end{definition}

\subsection{Security}
\label{sec:security-definition}

At any time slot $t$, a full node $\client$ holds a rollup ledger $\LOG^{\client}_t$, finalized by $\client$ at slot $t$, and the corresponding state $\st^\client_t$, whereas each light client $\client$ holds a state root $\stc^\client_t$, accepted by $\client$ at slot $t$.
Each parent chain node $\client$ online at slot $t$ also holds a parent chain ledger constructed using the \emph{finalized} and \emph{available} parent chain blocks in $\client$'s view at slot $t$. 

We provide two security definitions for rollup full nodes and light client respectively.
Upon obtaining the same ledger, the full nodes can iteratively process the transactions, and maintain the same rollup state over time.
Hence, security for the rollup full nodes is the same SMR safety and liveness definition for ledgers  given in~\cite{sleepy}, but applied to the rollup ledger instead of the full parent chain ledger:

\begin{definition}[Ledger Security~\cite{sleepy}]
\label{def:log-sec}
A rollup or SMR protocol is said to be \emph{secure} with latency $\Tfull$ for the full nodes if its ledgers $\LOG^{\client}_t$ satisfy;
\begin{itemize}
    \item \textbf{Safety.} For any slots $t$, $s$ and honest full nodes $i$, $j$ (online at these slots), either $\LOG^i_{t}$ is a prefix of $\LOG^j_{s}$ or vice versa.
    For any honest full node $i$ and times $s$, $t \geq s$, $\LOG^i_{s}$ is a prefix of $\LOG^i_{t}$.
    \item \textbf{Liveness.} If a transaction $\tx$ is input to an honest node at time $t$, then $\tx \in \LOG^j_{s}$ for all slots $s \geq t + \Tfull$ and all honest nodes $j$.
\end{itemize}
\end{definition}

In contrast, the state security below applies to both full nodes and light clients that hold a state root.
\begin{definition}[State Security]
\label{def:state-security}
A rollup is said to be \emph{secure} with latency $\Tlight$ for full nodes and light clients if the state roots $\stc^{\client}_t$ satisfy; 
\begin{itemize}
    \item \textbf{Safety.} For any slots $t$, $s$ and light clients or honest full nodes $i$, $j$ (awake at these slots), $\stc^i_t$ and $\stc^j_s$ satisfy the following guarantees:
    There exist rollup ledgers $\LOG_1$ and $\LOG_2$ such that $\stc^i_t = \left<\transition^*(\genesisstate,\LOG_1)\right>$, $\stc^j_s = \left<\transition^*(\genesisstate,\LOG_2)\right>$, and either $\LOG_1$ is a prefix of $\LOG_2$ or vice versa. If $i=j$ and $t \geq s$, $\LOG_2$ is a prefix of $\LOG_1$ (or vice versa).
    \item \textbf{Liveness.} If a transaction $\tx$ is input to an honest node at time $t$, then for any slot $s \geq t + \Tlight$ and light client or honest full node $i$ (awake at that slot), there exists a ledger $\LOG$ such that $\stc^i_s = \left<\transition^*(\genesisstate,\LOG)\right>$ and $\tx \in \LOG$.
\end{itemize}
\end{definition}

Via the binding property of the state roots (commitments), we assume that except with negligible probability in $\lambda$, no PPT adversary $\mathcal{A}$ can find two ledgers $\LOG$ and $\LOG'$ such that $\LOG \neq \LOG'$ and $\left<\transition^*(\genesisstate,\LOG)\right> = \left<\transition^*(\genesisstate,\LOG')\right>$.

\subsection{Accountability}
\label{sec:accountability}

Most rollups are built on top of PoS consensus protocols, where security is enforced by imposing financial punishments (\eg, slashing) on the protocol violators.
This notion is formalized by accountable safety: 
\begin{definition}[Accountable Safety for a SMR Protocol~\cite{aa-dilemma}]
\label{def:accountable-safety}
A SMR protocol is said to provide accountable safety with resilience $f$ if when there is a safety violation (per Definitions~\ref{def:log-sec}, ledger security), (i) at least $f$ adversarial nodes are irrefutably identified by a cryptographic proof as protocol violators, and (ii) no honest node is identified, with overwhelming probability.
Such a protocol is said to provide \emph{$f$-accountable-safety}.
\end{definition}
Accountable safety concerns with the \emph{aftermath} of a successful safety attack, and akin to a \emph{forensics} analysis, requires full nodes to construct a \emph{cryptographic proof} that will convince an external observer that a certain subset of the validators deviated from the protocol, thus are adversarial~\cite{forensics}.

We extend the notion of accountable safety to rollups as follows:
\begin{definition}[Accountable Safety for a Rollup]
\label{def:rollup-accountable-safety}
A rollup protocol is said to provide accountable safety with resilience $f$ if when there is a safety violation (per Definitions~\ref{def:log-sec} or~\ref{def:state-security}, ledger or state security), (i) at least $f$ adversarial parent chain nodes are irrefutably identified by a cryptographic proof as protocol violators, and (ii) no honest parent chain node is identified, with overwhelming probability.
Such a rollup protocol is said to provide \emph{$f$-accountable-safety}.
\end{definition}

\section{Rollup System}
\label{sec:rollup-system}

A rollup system consists of a parent chain (\eg Ethereum's beacon chain), rollup full nodes and light clients (\cf Figure~\ref{fig:merged-consensus}).
It typically has a special full node called the \emph{sequencer} that collects transactions from the users, and posts them to the parent chain~\cite{starknet,zksync,optimism,fuel,arbitrum}.
To reduce Ethereum's $\mathsf{calldata}$ cost, transactions are typically compressed and posted in batches.
The parent chain makes the rollup data public, thus enabling any full node to reconstruct the latest state without relying on the sequencer.
Similarly, censored rollup users can bypass the sequencer by posting their transactions directly to the parent chain.

Full nodes execute the rollup transactions in batches called the \emph{rollup blocks}.
Each block contains a transaction root that commits to the transactions within the block, and a state root that commits to the state corresponding to the transactions \emph{both} in the block and its prefix.
Rollups borrow security from the parent chain via \emph{merged consensus}~\cite{adler2020building,ethresearch-merged-consensus} (Section~\ref{sec:merged-consensus}), where full nodes use the transaction sequence on the parent chain to output a total order across the rollup blocks.
However, unlike full nodes, light clients do not download or process these transactions, instead asking the full nodes for the latest valid state.
Section~\ref{sec:rollup-bridges} explains the differences between enshrined and sovereign rollups in terms of how light clients learn about the state.
Section~\ref{sec:rollup-types} investigates optimistic and ZK based methods for proving state validity in different rollup types.


\subsection{Merged Consensus}
\label{sec:merged-consensus}

\begin{figure}[h]
    \centering
    \includegraphics[width=\linewidth]{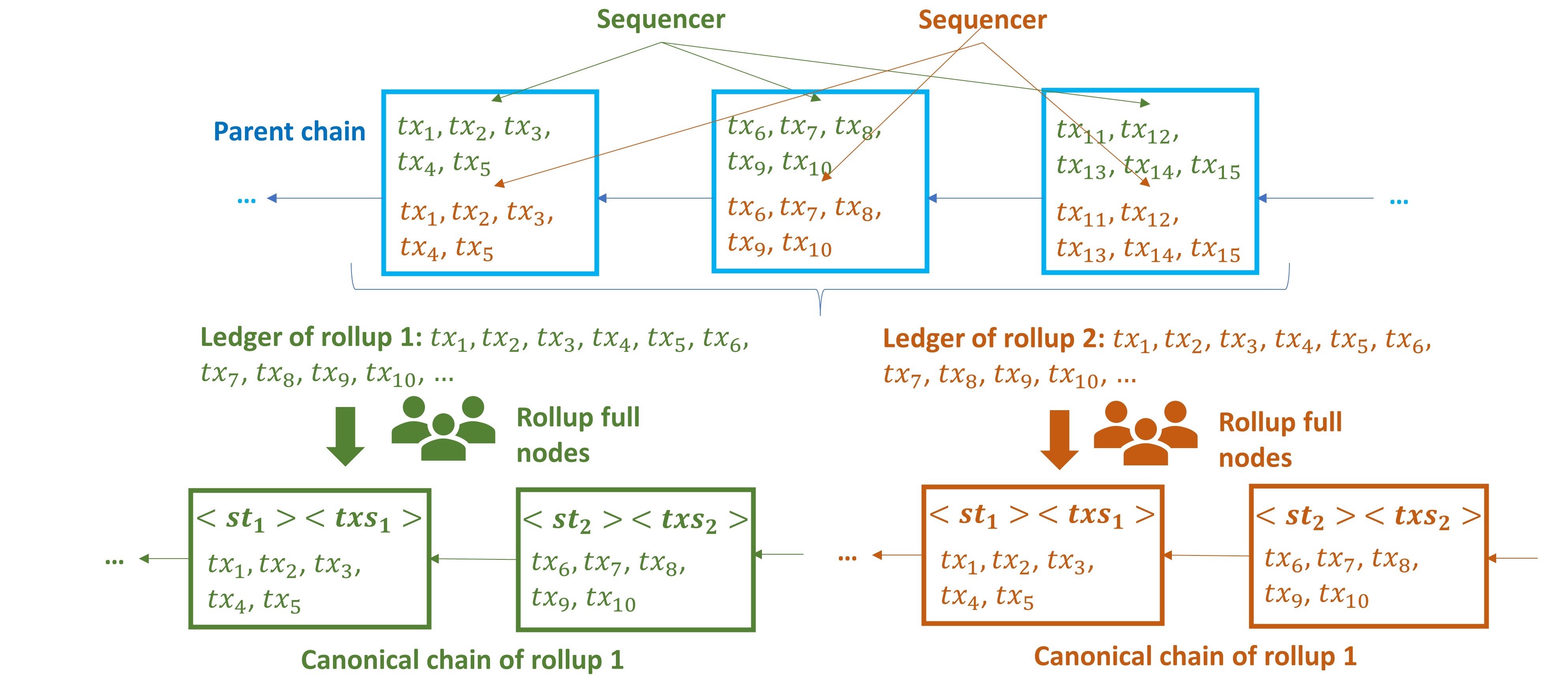}
    \vspace{-0.5cm}
    \caption{The figure depicts two rollups, colored green and brown, that use the same parent chain. Full nodes of each rollup download the transactions of their rollup (green and brown respectively), group them into rollup blocks, and execute them, obtaining the state and transaction roots.}
    \label{fig:merged-consensus}
\end{figure}

Merged consensus can be described by a validity function, fork-choice rule, finalization rule and a sequencer selection algorithm.
Since the sequencer selection algorithm is not crucial for security, but is an optimization for better liveness, we delegate its discussion to Appendix~\ref{sec:appendix-latency}.

\subsubsection{Validity.} A rollup block $b$ containing the transaction sequence $\txs$ and its state root are said to be \emph{valid} with respect to its prefix if there exists a state $\st$ such that $b$'s state root equals $\left<\transition^*(\st,\txs)\right>$, and the state root of the last rollup block in $b$'s prefix equals $\left<\st\right>$.
A chain of rollup blocks is valid if its blocks are valid.

\subsubsection{Fork-choice and Finalization Rules.}
To construct the \emph{canonical rollup chain}, an honest full node first downloads the rollup ledger, \ie, the sequence of rollup transactions on the parent chain (\cf Optimism and Arbitrum~\cite{optimism,arbitrum}).
It then executes these transactions in the specified order, and organizes them into a chain of rollup blocks, each with the corresponding state and transaction roots (\cf Figure~\ref{fig:merged-consensus}).
The full nodes are assumed to obtain the same sequence of transactions upon observing the same parent chain block (\eg, the first $1000$ bits of the data is reserved for the rollup).
Similarly, given the same transaction sequence, they are assumed to output the same rollup blocks. 
For example, in Optimism, the sequencer posts rollup transactions in batches, which can be directly translated into blocks.
A rollup block $b$ is said to be \emph{finalized} in the view of an honest full node $\client$ at slot $t$, if $b$ first appears as part of the canonical rollup chain in $\client$'s view at slot $t$.
Blocks on honest full nodes' canonical chains are valid.

It is possible that the sequencer first creates and executes the rollup blocks prior to posting their state roots and rollup data (such as state `deltas') to the parent chain (\cf, zkSync~\cite{zksync}).
Then, if there are invalid rollup blocks posted to the parent chain, the full nodes ignore these blocks while extracting the rollup ledger from the parent chain.


\subsection{Bridge to the Parent Chain}
\label{sec:rollup-bridges}

\begin{figure}[h]
    \centering
    \includegraphics[width=\linewidth]{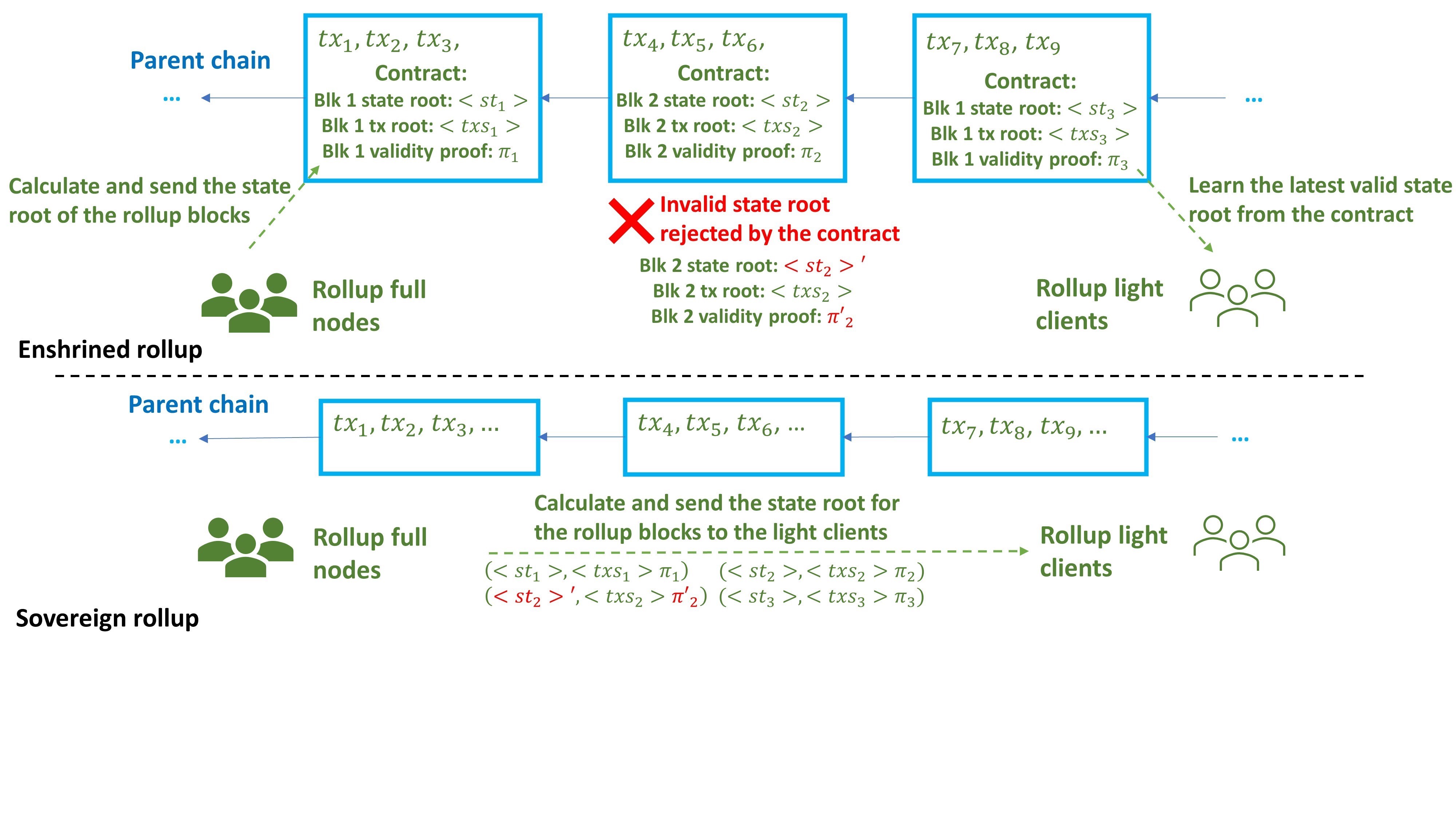}
    \vspace{-1.5cm}
    \caption{In an enshrined ZK rollup (top), the bridge contract checks the validity proofs $\pi_1, \pi_2, \ldots$ associated with the state root of the rollup blocks to detect invalid states such as $\left<\st_2\right>'$ (red). As this guarantees the validity of the root in the contract, light clients can directly obtain the latest root from the chain. In a sovereign ZK rollup (bottom), there is no contract on the parent chain, and the clients obtain the roots and proofs directly from the full nodes. An adversarial node can claim that $\left<\st_2\right>'$ is valid; however, upon observing the incorrect proof $\pi'_2$ (red), light clients invalidate it.}
    \label{fig:rollup-types}
\end{figure}

Depending on the parent chain's role in validating rollup blocks, rollups fall into two categories: \emph{enshrined} and \emph{sovereign}.

\subsubsection{Enshrined Rollups.} In an enshrined rollup, a \emph{bridge contract} executed by the parent chain nodes keeps track of the state root of the latest valid rollup block (but not the whole state).
The bridge contract continuously receives state and transaction roots of the new rollup blocks, and checks the validity of the state root with the help of the transaction root, and fraud (fault) or validity (ZK) proofs (\cf Figure~\ref{fig:rollup-types}).
If the check fails for the most recent root, the contract rejects the root, labelling the corresponding rollup block (if there is any) as invalid.
If the check succeeds, the contract accepts the roots as representing the most recent state and rollup block.
As the bridge contract validates the state root determined by the canonical rollup chain, light clients obtain the most recent valid root directly from the contract.
Thus, at any slot, clients can \emph{finalize} the latest valid root recorded by the contract.

Bridge contract allows rollup users to transfer money from the parent chain to the rollup by locking funds within the contract.
For withdrawals, users typically submit a burn transaction to the rollup.
Locked funds are released on the parent chain once the funds are burned in the rollup side, and a state root capturing this operation is accepted by the contract.

\subsubsection{Sovereign Rollups.} In a sovereign rollup, there is no bridge contract (\cf Figure~\ref{fig:rollup-types}).
The parent chain stays agnostic towards the rollup state, and is used only to ensure the ordering and availability of the rollup transactions as in Section~\ref{sec:merged-consensus}.
Thus, light clients receive state and transaction roots of the rollup blocks directly from the full nodes, and rely on the honest nodes to identify the latest valid state root.
Full nodes generate fraud proofs to warn the light clients about invalid roots, or validity proofs to convince them of a root's validity.
These proofs can be sent to the parent chain, or over the network to the clients.

Sovereign rollups handle deposits and withdrawals through cross-chain bridges as each sovereign rollup behave like an independent L1 blockchains towards other chains (\cf~\cite{layer-2-bridges} for an in-depth discussion of bridges).

\subsection{Proving Validity of the Rollup State}
\label{sec:rollup-types}

\subsubsection{Enshrined Optimistic Rollups.}
In these rollups, instead of verify the new state roots directly, the contract trusts honest full nodes called \emph{watch towers} to warn if a given root is invalid.

\paragraph{Non-interactive fraud proofs.}
To support fraud proofs, the full nodes calculate \emph{intermediate state roots (ISRs)} after every few transactions while executing the rollup transactions.
They commit to these ISRs as part of the rollup block, and post these roots to the parent chain with the final state and transaction root of the block~\cite{albassam2019fraud}.
If the contract receives an invalid state root for a claimed rollup block, a watch tower sends a fraud proof to the contract to point out the first invalid ISR within the block.
The proof consists of pointers to the last valid ISR, the first invalid one, the transaction in between, and the Merkle proofs for the state elements touched by the transaction with respect to the last valid ISR.
(The proof also contains Merkle proofs for the ISRs and the transaction with respect to the transaction root.)
The contract then verifies the Merkle proofs and calculates the next ISR by applying the transaction to the touched state elements.
It finally verifies the fraud proof by comparing the calculated root with the existing invalid ISR on the chain, their difference implying invalidity.

\paragraph{Refereed Games.}
In an enshrined rollup using refereed games~\cite{refereed} (\eg, Optimism~\cite{optimism}, Arbitrum~\cite{arbitrum}), a watch tower, \ie the challenger, demonstrates the invalidity of the state root of a disputed block by playing a \emph{bisection} game against another full node, \ie, the responder, that defends the root.
The game is mediated by the bridge contract that monitors the communication between the two parties.
At each step of the game, the responder divides the transaction sequence within the disputed block into $K$ (\eg, $2$) \emph{pieces} and calculates the ISRs at their ends.
Then, the challenger claims one of these roots to be invalid, after which the game is repeated on the chosen piece with the roles reversed for the players.
The bisection continues until both parties converge on a single simple instruction whose execution is the source of the disagreement.
Finally, the bridge contract succinctly checks whether processing this instruction outputs the ISR claimed by any party using the same method as the fraud proofs.  
The truthful party wins the game, its claim accepted by the verifier.

\paragraph{Succinctness.} To avoid spamming attacks, where adversarial full nodes send many invalid state roots for the same rollup block, optimistic rollups impose financial punishments for posting invalid roots, and challenging valid ones with fraud proofs or bisection games.
The full nodes can be required to put up stake in the contract before sending roots or fraud proofs, at the peril of losing it if they misbehave.
Assuming the succinctness of verifying fraud proofs, if the number of invalid roots claimed per rollup block is bounded, the contract's communication and computational complexity per block remains succinct in the block size.

\paragraph{Delay of Finality.} In an optimistic enshrined rollup, a state root received by the contract remains in a pending state until sufficient time passes (which we denote by $\Tfraud$), during which no fraud proof or bisection game invalidates the block. 
Otherwise, they are deemed invalid.
This is to give watch towers enough time to challenge invalid roots, which is further discussed in Appendix~\ref{sec:appendix-latency}.

\subsubsection{Sovereign Optimistic Rollups.}

\paragraph{Non-interactive fraud proofs.}
In a sovereign rollup using non-interactive fraud proofs, the full nodes would again post the ISRs to the parent chain to ensure that all light clients observe and agree on them.
If a state root obtained by a client is invalid, a watch tower can send a fraud proof to the client over the network.
As there is no contract to verify the proof, the light clients themselves verify the proofs using the ISRs and transactions on the chain.
Although this can be done succinctly for one fraud proof, a light client observing the protocol at some slot $r$ might have to download and verify as many fraud proofs as the number of rollup blocks posted until slot $r$; since there is no contract that can be trusted to have verified past proofs and found the previous valid state (\cf~\cite{woods-attack} and~\cite[Appendix C]{lazylight}).
Since such a light client has complexity $\Theta(r)$, it is not possible to support succinct clients on sovereign rollups using fraud proofs the same way they are used on enshrined rollups. 

\paragraph{Refereed Games.}
Tas et al.~\cite{lazylight} proposed the first succinct light client protocol for a sovereign rollup based on refereed games.
Their light client, \ie, the verifier, connects to multiple rollup full nodes, \ie, the provers (one of which is assumed honest), to get the latest valid state root.
The honest provers update and maintain a Merkle mountain range~\cite{mmr,mmr-grin} over the sequence of rollup transactions received by the parent chain, and the corresponding ISRs.   
If the provers disagree about the latest root, the verifier mediates a tournament of bisection games among them to discover the truthful parties with the same valid root.
In each game, the verifier walks down the Merkle trees of the disagreeing parties and identifies the first point of disagreement between the ISR sequences of the two parties.
Finally, the verifier finds the party with the correct ISR by processing a single transaction and using the same method as the fraud proofs.

Since the games are played across all rollup transactions on the parent chain at each slot $r$, the verifier overcomes the problem of validating historic invalid states, which ensures the succinctness of the light clients.
However, these games might be more susceptible to eclipse attacks since the security of the construction requires one honest prover connected to \emph{each} light client, rather than a single watch tower monitoring the parent chain.

\paragraph{Succinctness.} To discourage spamming attacks, sovereign rollups can again adopt slashing rules for the full nodes sending invalid state roots (signed by the nodes) to the light clients as part of its transition function, and burn the native tokens of malicious nodes within the rollup state.

\subsubsection{Zero-Knowledge (ZK) Rollups}

In a ZK rollup, state and transaction roots of rollup blocks are posted to the contract along with succinctly verifiable, short validity (ZK) proofs generated by the full nodes using verifiable computation techniques such as Groth~\cite{groth}, Plonk~\cite{plonk} and STARKs~\cite{stark}.

\paragraph{Enshrined ZK Rollups.}
The bridge contract accepts a claimed state root for a new rollup block only after verifying its validity proof.
During verification, the contract uses as public inputs, the verified state root of the parent block, the claimed state and transaction roots for the new block, and select data from the new rollup block (assumed to be on the chain).
State roots without a correct validity proof are immediately rejected (\cf Figure~\ref{fig:rollup-types}).
Assuming that the number of invalid roots submitted to the contract per rollup block is bounded (no spamming) and the proof size and its verification time are succinct, the complexity of the contract per block remains succinct in the block size.

\paragraph{Sovereign ZK Rollups.}
In a sovereign rollup, a bootstrapping light client again faces the issue of verifying historic invalid state roots generated since the beginning of the protocol.
Although the client has to verify only a succinct validity proof to detect an invalid root, for a client online at slot $r$, there might be $\Theta(r)$ many roots, implying linear complexity in the protocol runtime.
(An enshrined rollup solves this problem using the contract that verifies old state and transaction roots.)
To overcome this problem, sovereign rollups can use recursive SNARKs~\cite{recursiveSNARK,sasson-recursive-snark} such as Mina~\cite{coda} or Halo~\cite{halo,bunz-recursive-proof} with succinctly verifiable proofs.
Another alternative is to use Plumo~\cite{plumo} that can prove four months of blockchain state history with a single transition proof, in practice significantly reducing the complexity of verifying old proofs.

With recursive constructions, the light client has to verify only a single proof generated by the full nodes to validate the \emph{whole} blockchain history summarized by the latest state root.
Thus, without spamming, a sovereign ZK rollup can achieve succinctness via recursive SNARKs or STARKs.
To support recursive SNARKs, full nodes maintain a root (commitment) of the whole transaction history including those in the latest rollup block.
During verification the light client uses as public inputs, the new claimed state root, the new claimed root for the transaction history, and select data from among these transactions.

\section{Designing Rollups with Accountable Safety}
\label{sec:accountable-rollup}

\begin{figure}[h]
    \centering
    \includegraphics[width=1.1\linewidth]{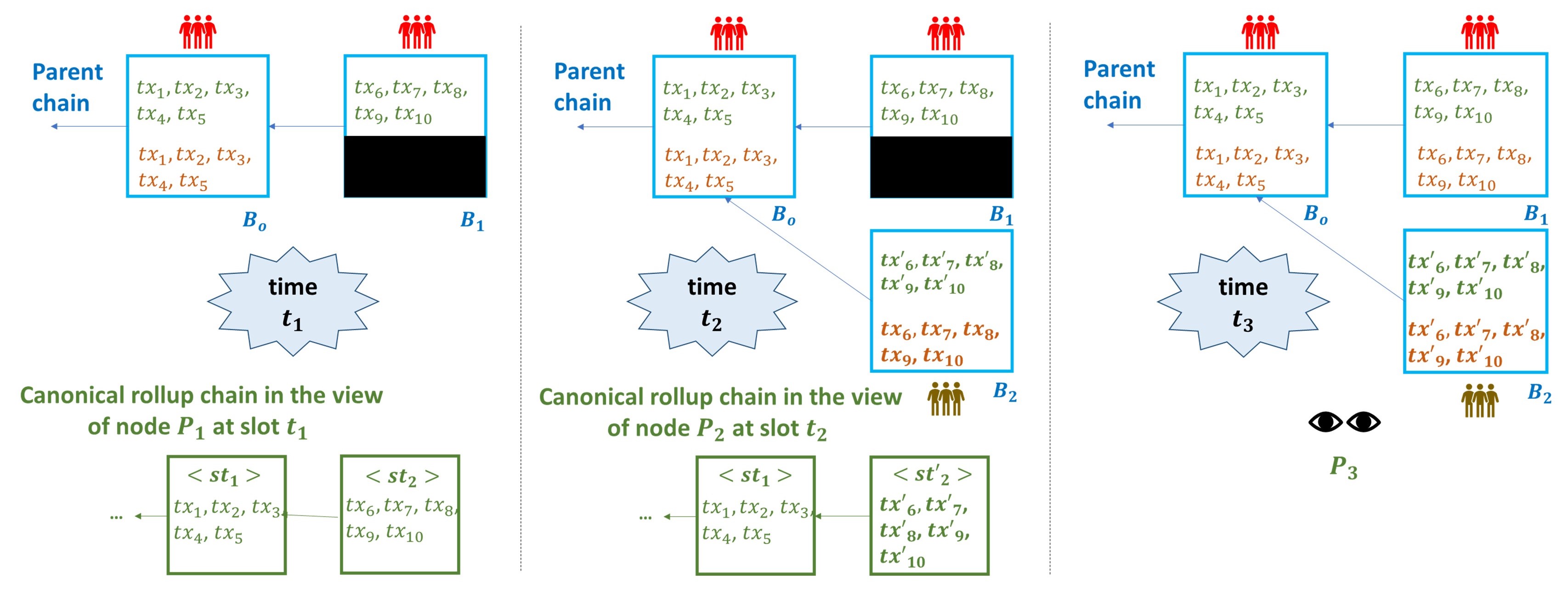}
    \caption{Parts of the data within block $B_1$, denoted black, are withheld from the rest of the network by the adversarial parent chain nodes, shown in red, that have finalized $B_1$. At slot $t_1$, node $P_1$ outputs a chain containing the rollup transactions within $B_1$. At slot $t_2$, $P_2$'s rollup chain does not contain the transactions within $B_1$ since the social consensus of the parent chain has restarted it from block $B_2$. At slot $t_3$, the unavailable data within $B_1$ is released, making it impossible for a late-joining observer $P_3$ to irrefutably accuse any parent chain node of protocol violation. The attack works irrespective of whether the rollup is an enshrined or sovereign rollup.}
    \label{fig:accountability-attack}
\end{figure}

\subsection{A Simple Attack on Accountability.}

The rollup constructions described so far do not provide accountable safety, even if they are built on an accountably-safe parent chain.
We show this with a safety attack, after which no adversarial parent chain node can be provably accused.

Consider a ZK rollup built on PoS Ethereum, and suppose the adversary temporarily controls over $2/3$ of the parent chain nodes, called the validators (colored red on Figure~\ref{fig:accountability-attack}).
With their votes, they propose and finalize a valid parent chain block, $B_1$, at slot $t_1$, yet withhold parts of its data from the rest of the network.
However, they make all the transactions of the green rollup available, enabling the full nodes to download them and advance the rollup state.
At this point, rollup nodes are not aware of the \emph{data availability attack}, and cannot detect that an attack is happening; as they do not download non-rollup transactions (\cf Section~\ref{sec:nodes-overview}).
Hence, the rollup chain of a full node $P_1$ at slot $t_1$ contains the rollup transactions within the unavailable block (\cf Figure~\ref{fig:accountability-attack}, left).
Note that it is crucial for the adversarial validators to hide only the non-rollup transactions in $B_1$.
For instance, if the first $100$ bits of $B_1$ are reserved for the rollup transactions, adversary should \emph{not} withhold any of these first $100$ bits.
Otherwise, the rollup nodes would become aware of the attack, and $P_1$ could refuse to output a rollup chain containing even the available rollup transactions from $B_1$, thus thwarting the attack.

Since the unavailability of the withheld transactions implies censorship, the honest validators subsequently engage in extra-protocol consensus \emph{on the social layer} to create a `minority soft fork', slashing the stake of the adversarial validators, and inaugurating a new, honest validator set (golden on Figure~\ref{fig:accountability-attack})~\cite{ethresearch-51-attack}.
Thus, Pos Ethereum is restarted with fresh validators from the last \emph{available} block, $B_0$, after which the validators and rollup nodes download only the blocks built by the new validators.
At slot $t_2>t_1$, the validators finalize a block $B_2$ that essentially conflicts with $B_1$.
Hence, the rollup chain of a node $P_2$ at slot $t_2$ includes transactions from the new block, and conflicts with the rollup chain of $P_1$ at slot $t_1$, implying a safety violation per Definition~\ref{def:log-sec} (\cf Figure~\ref{fig:accountability-attack}, center).

At slot $t_3$, the adversarial validators publicize the hidden data within the unavailable block (\cf Figure~\ref{fig:accountability-attack}, right).
Suppose an external observer $P_3$ joins the protocol at slot $t_3$.
At this point, not only the transactions within both forks of PoS Ethereum are available and valid in $P_3$'s view, but the two forks have different validator sets, making it impossible for $P_3$ to irrefutably accuse any validator of protocol violation such as double-voting for the conflicting blocks $B_1$ and $B_2$.
Hence, no validator, even those who witnessed the data availability attack can generate a cryptographic proof that will convince $P_3$, who has not observed the attack, that the red validators are adversarial.
This implies the lack of accountable safety.

\subsection{An Accountably-Safe Rollup.}

The attack above highlights the main problem preventing accountability: the rollup nodes do not verify the \emph{availability} of non-rollup transactions within the parent chain blocks.
To mitigate this problem, they must use data availability sampling (DAS), through which they can validate the availability of the blocks without downloading them.
For instance, if $P_1$ used DAS to verify $B_1$'s availability, it would have noticed the data availability attack and refused to download the transactions in $B_1$, avoiding any future conflict with $P_2$.
For illustration, we describe an accountably-safe sovereign ZK rollup below, based on an accountably-safe parent chain.

Both rollup full nodes and light clients validate the finality of the parent chain blocks by downloading their headers.
They also do DAS on these blocks to verify their availability (Figure~\ref{fig:rollup-diagram}).
Then, full nodes download rollup transactions from the available and finalized blocks, and construct the rollup chain as described in Section~\ref{sec:merged-consensus}.
They also update and maintain a root of the transaction history, the current state corresponding to these transactions and the latest validity proof found via the recursive SNARK.
The light clients query the full nodes for the state root of the last rollup block, upon which the honest full nodes reply with the queried root, the new claimed root of the transaction history and the validity proof.
The proof might require some select data from the transaction history, assumed to be recorded on the parent chain (\cf Section~\ref{sec:rollup-types}).
Upon verifying the proof, the light client accepts the given state root as the most recent valid root.


\section{Security}
\label{sec:security}

To prove security, we first give definitions for the completeness and soundness of the fraud proofs, refereed games and the proof systems.
A fraud proof scheme is \emph{complete} if a proof generated by the honest watch tower, accusing an invalid state root for the latest rollup block, always convinces the contract or the light clients of the root's invalidity.
Similarly, a refereed game is complete if the honest watch tower accusing an invalid state root is always found truthful by the verifying contract or client.
A fraud proof scheme is \emph{sound} if no PPT adversary $\mathcal{A}$ can generate a verified fraud proof accusing a valid state root except with negligible probability in the security parameter $\lambda$.
Similarly, a refereed game is sound if no PPT adversary $\mathcal{A}$ accusing a valid state root is found truthful by the verifying contract or the light clients except with negligible probability.
These properties are typically proven using the security of the underlying primitives such as binding vector commitments.
A proof system is \emph{complete} if the verifier always accepts the proof for a valid state root.
It is \emph{sound} if no PPT adversary $\mathcal{A}$ can generate a proof that is accepted by an honest verifier except with negligible probability.

We next state the \emph{respective} assumptions for the different rollup types, which will be used in the subsequent proofs.
\begin{itemize}
\item \textbf{Enshrined Optimistic Rollup:} The network is $\Delta$ synchronous, there is always an honest watch tower, the parent chain is live with some latency $T$, and the fraud proofs or the refereed games satisfy security (\ie, completeness and soundness).
\item \textbf{Sovereign Optimistic Rollup:} The network is $\Delta$ synchronous, there is an honest watch tower connected to every light client, and the refereed games satisfy security.
\item \textbf{Enshrined ZK Rollup:} The proof system is secure (complete and sound).
\item \textbf{Sovereign ZK Rollup:} The recursive proof system is secure.
\end{itemize}

\begin{theorem}[Accountable Safety for Rollups]
\label{thm:acc-safety}
Suppose the parent chain satisfies $f$-accountable safety and there exists a secure data availability sampling (DAS) scheme for the parent chain.
Then, each rollup type satisfies $f$-accountable safety for both full nodes and light clients under the respective assumptions.
\end{theorem}

Proof of Theorem~\ref{thm:acc-safety} depends on the following lemmas:

\begin{lemma}[State Validity for Enshrined Optimistic Rollups]
\label{lem:validity-enshrined-optimistic}
Suppose the network is synchronous, there is an honest watch tower at all times, the parent chain satisfies liveness with latency $T$, and the fraud proofs or the refereed games are complete and sound. 
Then, if a claimed state root for a new rollup block is accepted by the bridge contract, with overwhelming probability, it is the valid state root corresponding to the sequence of transactions within the latest rollup block and its prefix.
\end{lemma}

\begin{proof}
If a state root claimed for a new rollup block and input to the contract is invalid, an honest watch tower immediately creates a fraud proof or challenges the new root with a refereed game (\cf Section~\ref{sec:rollup-types}).
The fraud proof is then received and verified by the contract within $T$ slots.
Similarly, the game terminates after at most some $T_{\mathrm{game}}$ slots, with the contract verifying the instruction at the source of disagreement.
By the completeness of the fraud proofs and the refereed games, the fraud proof convinces the contract of the root's invalidity, and the honest watch tower is found truthful by the contract.
Recall that the contract of an optimistic rollup waits for $\Tfraud$ slots upon receiving the claimed state and transaction roots for a new rollup block.
Thus, by setting $\Tfraud > T, T_{\mathrm{game}}$, the contract ensures that every accepted state root is valid with respect to the transactions within the latest rollup block and its prefix with overwhelming probability.
\end{proof}

\begin{lemma}[State Validity for Sovereign Optimistic Rollups]
\label{lem:validity-sovereign-optimistic}
Suppose the network is synchronous, there is an honest watch tower at all times, and the refereed games of the sovereign optimistic rollup is complete and sound.
Then, if a state root claimed by a full node for a new rollup block is accepted by a light client, with overwhelming probability, it is the valid state root corresponding to the sequence of transactions within the latest rollup block and its prefix.
\end{lemma}

Since the only known light client protocol for sovereign optimistic rollups is given by \cite{lazylight}, validity of accepted state roots follows from \cite[Theorem 4]{lazylight} under the given assumptions.

\begin{lemma}[State Validity for Enshrined ZK Rollups]
\label{lem:validity-enshrined-zk}
Suppose the proof system used by the ZK rollup is complete and sound.
Then, if a claimed state root for a new rollup block is accepted by a light client, it is the valid state root corresponding to the sequence of transactions within the latest rollup block and its prefix.
\end{lemma}

\begin{proof}
By the soundness of the proof system, no PPT adversary can find a validity proof, except with negligible probability, such that the contract will output true given an invalid state root, the associated transaction root, the valid state root of the previous block (that was already verified by the contract), and the rollup data on the parent chain claimed to be committed by the transaction root.
Hence, the contract ensures that every accepted state root is valid with respect to the transactions within the latest rollup block and its prefix with overwhelming probability.
\end{proof}

\begin{lemma}[State Validity for Sovereign ZK Rollups]
\label{lem:validity-sovereign-zk}
Suppose the recursive proof system used by the ZK rollup is complete and sound.
Then, if a claimed state root for a new rollup block is accepted by a light client, it is the valid state root corresponding to the sequence of transactions within the latest rollup block and its prefix with overwhelming probability.
\end{lemma}

\begin{proof}
By the soundness of the proof system, no PPT adversary can find a validity proof except with negligible probability that will output true given an invalid state root, the associated root of the transaction history, and the rollup data on the parent chain claimed to be committed by the root of transaction history.
Hence, the contract ensures that every accepted state root is valid with respect to the transactions within the latest rollup block and its prefix with overwhelming probability.
\end{proof}

\begin{proof}[Proof of Theorem~\ref{thm:acc-safety}]
Suppose there is a safety violation on the rollup such that for two slots $t$ and $s$, and light clients (or full nodes) $i$ and $j$ awake at the respective slots, it holds that $\stc^i_t = \left<\transition^*(\genesisstate,\LOG_1)\right>$ for some rollup ledger $\LOG_1$, $\stc^j_s = \left<\transition^*(\genesisstate,\LOG_2)\right>$ for some rollup ledger $\LOG_2$, and $\LOG_1$ conflicts with $\LOG_2$.
Let $\ch_1$ and $\ch_2$ denote the parent header chains in $i$ and $j$'s views at slots $t$ and $s$.

Suppose $i$ and $j$ are full nodes.
Then, $\LOG_1$ and $\LOG_2$ are the rollup ledgers in $i$ and $j$'s views at slots $t$ and $s$, downloaded from the chains $\ch_1$ and $\ch_2$ respectively.
As these ledgers conflict with each other and full nodes download the same sequence of rollup transactions from the same parent chain block, $\ch_1$ and $\ch_2$ must also be conflicting with each other.
Now, since $i$ and $j$ check the finality of the parent chain blocks, and verify their availability via DAS, by the soundness of DAS, all parent chain blocks on $\ch_1$ in $i$'s view at slot $t$ and those on $\ch_2$ in $j$'s view at slot $s$ are finalized and valid.
However, this implies a safety violation on the parent chain per Definition~\ref{def:log-sec}.
As the parent chain provides accountable safety with resilience $f$, in this case, at least $f$ adversarial parent chain nodes are identified by a cryptographic proof as protocol violators (and no honest node is identified) with overwhelming probability.

Suppose $i$ and $j$ are light clients.
If the rollup is an enshrined rollup, $i$ and $j$ must have received $\stc^i_t$ and $\stc^j_s$ from the contract as determined by the chains $\ch_1$ and $\ch_2$ at slots $t$ and $s$.
By Lemmas~\ref{lem:validity-enshrined-optimistic} and~\ref{lem:validity-enshrined-zk}, under the respective assumptions for the enshrined rollups, the roots accepted by the contract are the valid state roots corresponding to the sequence of transactions within the latest rollup block and its prefix with overwhelming probability.
On the other hand, if the rollup is a sovereign rollup, $i$ and $j$ must have received $\stc^i_t$ and $\stc^j_s$ from the full nodes at slots $t$ and $s$.
By Lemmas~\ref{lem:validity-sovereign-optimistic} and~\ref{lem:validity-sovereign-zk}, under the respective assumptions for the optimistic rollups, the state roots accepted by the clients are the valid state roots corresponding to the sequence of transactions within the latest rollup block and its prefix.
Since the ordering of the rollup blocks is determined by the ordering of the rollup transactions on the parent chain, $\LOG_1$ and $\LOG_2$ must be the prefix of or the same as the rollup ledgers ingrained in $\ch_1$ and $\ch_2$ in both cases.
Then, by the same argument as the full nodes, $\ch_1$ and $\ch_2$ conflict with each other, and at least $f$ adversarial parent chain nodes are identified by a cryptographic proof as protocol violators (and no honest node is identified) with overwhelming probability.
\end{proof}

\begin{theorem}[Liveness for Rollups]
\label{thm:liveness}
Suppose the parent chain satisfies safety and liveness with some latency $T$.
Then, each rollup type satisfies liveness with some latency $\Tconf$ for both full nodes and light clients under the respective assumptions, with overwhelming probability.
\end{theorem}

\begin{proof}[Proof of Theorem~\ref{thm:liveness}]
Since the parent chain is safe and live with latency $T$, any rollup transaction $\tx$ sent by a user (\eg, sequencer or client) at slot $t$ will be included and stay in the ledger of all honest parent chain nodes by slot $t+T$.
Thus, it will be downloaded and finalized as part of the rollup chain by all honest rollup full nodes by slot $t+T$.
Afterwards, if the rollup is an enshrined rollup, an honest full node will submit the valid state and transaction roots of the new rollup block containing $\tx$ to the bridge contract.
These roots will be received by the contract by slot $t+2T$.
Similarly, if the rollup is a sovereign rollup, an honest full node will send the valid state and transaction roots of the new rollup block containing $\tx$ to the light clients over the network, which will be received by slot $t+T+\Delta$.

If the rollup is an enshrined optimistic rollup, by the soundness of the fraud proofs, no PPT adversary will be able to generate a fraud proof that will convince the contract that the state root is invalid, except with negligible probability in $\lambda$.
Similarly, by the soundness of the refereed games, no PPT adversary accusing the valid state root will be found truthful by the contract except with negligible probability.
Hence, by slot $t+2T+\Tfraud$, the state root of the new rollup block containing $\tx$ will be accepted by the contract and the light clients.
If the rollup is an enshrined ZK rollup, the full node sending the state root will attach a validity proof that will be verified by the contract due to the completeness of the proof system.
Hence, the state root of the new rollup block containing $\tx$ will be accepted by the contract and the light clients at slot $t+2T$.

If the rollup is a sovereign optimistic rollup, by \cite[Theorems 3,4]{lazylight}, the light client will accept the state root handed by the honest full node by slot $t+T+\Delta+\Tfraud$.
If the rollup is an sovereign ZK rollup, the full node sending the state root will attach a validity proof that will be verified by the contract with respect to the transaction history of the rollup due to the completeness of the recursive proof system.
Hence, the state root of the new rollup block containing $\tx$ will be accepted by the light client at slot $t+T+\Delta$.
Consequently, if $\tx$ is input to an honest node at time $t$, which sends it to the parent chain, then for any slot $s \geq t + \Tlight$, where $\Tlight \geq t+2T+\Tfraud+\Delta$, for any light client or honest full node $i$ awake at slot $s$, there exists a ledger $\LOG$ such that $\stc^i_s = \left<\transition^*(\genesisstate,\LOG)\right>$ and $\tx \in \LOG$, implying the liveness of the rollup.
\end{proof}


\section{Analysis}
\label{sec:analysis}

\subsection{Comparison of Security between Optimistic and ZK Rollups}
\label{sec:comparison-security}

All rollups require the liveness, more specifically, censorship-resistance of the parent chain for the liveness of the rollup.
However, they differ in their assumptions for safety of the rollup, besides requiring the safety of the parent chain: while the validity proofs of ZK rollup ensure the safety of the state as long as the proof system is complete and sound (\ie, cryptographic assumptions), optimistic rollups require both the existence of an honest watch tower as well as the liveness of the parent chain.

An oft-overlooked issue with the assumption of an online watch tower is the denial-of-service (DoS) risk for these nodes.
Without network anonymity, attackers can observe the transactions downloaded by the nodes and identify those downloading the rollup specific ones as watch towers.
This makes DoS resistance is an inherent assumption for the availability of watch towers on optimistic rollups.
Appendix~\ref{sec:appendix-incentives} gives an in-depth discussion of the supportive and disruptive roles incentives play regarding the assumptions for the safety of optimistic rollups.

As for ZK rollups, their security often depend on the security of the trusted setup.
Whereas Groth 16'~\cite{groth} requires a trusted setup with at least one honest participant for each new smart contract, Sonic~\cite{sonic}, Plonk~\cite{plonk} and Marlin~\cite{marlin} introduce a single universal setup, which can be updated for new smart contracts.
Finally, STARKs~\cite{stark}, Halo~\cite{halo}, Supersonic~\cite{supersonic} and Fractal~\cite{fractal} remove the trusted setup requirement altogether, eliminating the need for an honest participant.
Another concern for ZK rollups is the soundness of the novel cryptographic assumptions underlying the verifiable computation techniques as opposed to the more battle-tested cryptography (\eg, cryptographic hash functions) of optimistic rollups.

\subsection{Throughput and Scalability}
\label{sec:analysis-throughput}

We next formalize throughput and scalability for rollups.
Let $g_\transition(\st,\tx)$ denote the gas cost, more generally, the \emph{cost function} of applying a transaction $\tx$ to a state $\st$ using the transition function $\delta$.
If a rollup full node executes transactions at a \emph{processing speed} of $C$ gas per time slot, it applies $\tx$ to $\st$ in $g_\transition(\st,\tx)/C$ expected slots.
The processing speed quantifies the cost of running a full node, and its value depends on the node's processor speed, the upload and download rates of its I/O ports, or other resource bottlenecks.

Consider a sequence of rollup transactions $\tx_1,\tx_2,\ldots,\tx_n$ of the same type and the states $\st_i:=\transition(\st_{i-1},\tx_i)$, $i=1,..,n$.
Suppose each full node has the same processing speed $C$.
Let $\Th$ denote the throughput of the parent chain measured for the rollup transactions.
\begin{definition}
\label{def:throughput}
The rollup throughput is defined as the number of transactions processed by the full nodes per unit time:
\begin{IEEEeqnarray*}{C}
\label{eq:throughput}
\lim_{n \to \infty} \min\left(\frac{n}{\sum_{i=1}^n g_\transition(\st_{i-1},\tx_i)/C}, \Th \right) = \lim_{n \to \infty} \min\left(C \cdot \frac{n}{\sum_{i=1}^n g_\transition(\st_{i-1},\tx_i)}, \Th \right)
\end{IEEEeqnarray*}
where $g_\transition(\st_{i-1},\tx_i)/C$ is the expected time to process $\tx_i$.
\end{definition}

In practice, the processing speed of the Ethereum nodes is often restricted by execution (\eg, the rate of SSD accesses on commodity hardware) rather than consensus or data availability for the transactions on the parent chain~\cite{vitalik-tweet-1,vitalik-tweet-2}.
Moreover, removing the consensus bottleneck on throughput has been studied extensively, leading to the development of SMR protocols~\cite{bitcoin-ng,prism-ccs,narwhal-tusk} that achieve high transaction throughput without sacrificing latency or decentralization.
Hence, in the subsequent analysis, we assume that throughput is primarily bottlenecked by execution, and equals $\lim_{n \to \infty} C \cdot (n / \sum_{i=1}^n g_\transition(\st_{i-1},\tx_i))$.

The expression above can be further partitioned into two components: a \emph{vertical scalability} term $C$, and a \emph{horizontal scalability} term $n/\sum_{i=1}^n g_\transition(\st_{i-1},\tx_i)$.
Here, horizontal scalability refers to increasing throughput while keeping the processing speed constant, and is the main improvement enabled by rollups.
\begin{definition}
\label{def:scalability}
Horizontal scalability (scalability for short) is defined as throughput over processing speed:
\begin{IEEEeqnarray*}{C}
\lim_{n \to \infty} \frac{n}{\sum_{i=1}^n g_\transition(\st_{i-1},\tx_i)}.
\end{IEEEeqnarray*}
\end{definition}
Although the blockchains can require higher processing speeds $C$ for higher throughput, enforcing a large $C$ might prevent parties with commodity hardware from running full nodes.
This, in turn, results in the centralization of execution, requiring more trust in a select number of nodes and potentially giving them leverage over other system participants~\cite{blog-oru}.
By improving the horizontal scalability term, rollups enhance throughput without causing centralization.

\paragraph{Removing the Scalability Bottleneck.}
There are two ways rollups increase scalability:
(i) using a system with a smaller cost function, and
(ii) sharding the state into multiple rollups.

Rollups allow the use of execution models better suited to specific applications.
For instance, as UTXO-based execution minimizes SSD access for full nodes, it is cheaper to use a UTXO model for a rollup that is solely for payments.
This leads to a smaller cost function, $g_1(.)$, for the rollup compared to the function $g_0(.)$ of the parent chain: $g_1(\st_{i-1}, \tx_i) \leq g_0(\st_{i-1}, \tx_i)$.
Different execution model also provides a platform to experiment with algorithmic improvements such as stateless clients~\cite{ethereum-stateless}, and better parallelization of execution~\cite{github-easy-parallelizability}.

Rollups also enable full nodes to process only those transactions relevant to their application~\cite{adler2020building}.
Suppose there are $m$ rollups with the transactions $\tx^j_i$, where $j=1,..,m$ denotes the id of the relevant rollup.
Rollup $j$ has the cost function $g_j$ and the states $\st^j_i$.
These rollups share the same parent chain with the cost function and states denoted by $g_0(.)$ and $\st^0_i$ respectively, where we assume $g_j(\st^j_{i-1}, \tx^j_i) \leq g_0(\st^0_{i-1}, \tx^j_i)$ for all $j=1,\ldots,m$ as argued above.
Then, by sharding the execution into these $m$ rollups, scalability can be improved from
\begin{IEEEeqnarray*}{C}
\lim_{n \to \infty} \frac{n \cdot m}{\sum_{i=1}^n \sum_{j=1}^m g_0(\st^j_{i-1},\tx^j_i)} \text{ to } \lim_{n \to \infty} \frac{n \cdot m}{\sum_{i=1}^n \max_{j=1,..,m} g_j(\st^j_{i-1},\tx^j_i)},
\end{IEEEeqnarray*}
with a factor of at least $m$.
Note that the expression $\max_{j=1,..,m} g_j(\st^j_{i-1},\tx^j_i)$ is due to the rollup nodes executing \emph{only} their rollup transactions, thus making the slowest rollup the new bottleneck of the system, whereas without rollups, \emph{all} nodes processed \emph{all} transactions on the parent chain. 

Despite its performance benefits, sharding of execution into rollups comes with the loss of composability among applications.
To mitigate this, full nodes can run light clients of other rollups to learn relevant parts of the other states.

\subsection{Latency}
\label{sec:appendix-latency}

Suppose the parent chain is secure with latency $T$ per Definition~\ref{def:log-sec}.
Then, once a rollup transaction is input an honest full node (\eg, the sequencer) that submits it to the parent chain, it appears in the rollup ledgers of all honest full nodes within $T$ slots.
Assuming that the execution time is negligible compared to $T$, latency for rollup full nodes is thus determined by the latency of the parent chain.

In an optimistic enshrined rollup, the state roots received by the contract remain in a pending state until sufficient time passes, during which no fraud proof or bisection game invalidates the block (otherwise, they are deemed invalid).
We assume that whenever there is one honest watch tower and the network is synchronous, the contract and the light clients learn the invalidity of a root posted to the contract within $\Tfraud$ slots.
The exact value of $\Tfraud$ depends on the delay $\Delta$ of the network, the finalization latency of the parent chain and the availability of watch towers.
Consequently, to ensure state security per Definition~\ref{def:state-security}, light clients must wait for the $\Tfraud$ window before finalizing a state root on the contract (making the total latency approximately $T+\Tfraud$ for light clients on an optimistic enshrined rollup).
As $\Tfraud$ scales in the finalization latency of the parent chain, it should be large enough to grant the community enough time to resolve the attacks or bugs that cause censorship of fraud proofs.
Thus, its estimates tend to be on the order of weeks (\eg, one week on Arbitrum~\cite{arbitrum-delay}) despite the fact that fraud proofs can be created and included within the parent chain faster under normal circumstances.

For sovereign optimistic rollup, assuming a $\Delta$ synchronous network and the presence of at least one honest prover among the peers of every light client, we again denote the latency of the refereed game by $\Tfraud$~\cite{lazylight}.
Here, $\Tfraud$ need not depend on the finalization latency of the parent chain since the game is played over the network, enabling latency much less than a week (\cf the paper~\cite{lazylight} for the numbers).
However, it might be more susceptible to network level and eclipse attacks as argued in Section~\ref{sec:rollup-types}.

Recall that a client withdrawing from the rollup waits until the root of the rollup block containing its burn transaction is validated by the bridge contract (\cf Section~\ref{sec:rollup-bridges}).
As this implies a withdrawal delay of $\Tfraud$ in an optimistic rollup, light clients with urgent withdrawal requests might utilize liquidity providers (\eg, Hop~\cite{hop}, Connext~\cite{connext}, Dai Bridge~\cite{dai}) to atomically withdraw or transfer their funds~\cite{ethresearch-side-chain-halting,article-benjamin-simon}.

For ZK rollups, assuming that the proof size and its verification time are succinct, there is no pending root in the bridge contract or long delays before accepting these roots.
Hence, without the delay of executing the rollup transactions and calculating the validity proofs, the latency for the light clients matches those of the full nodes, \ie, $\sim T$.
However, the fixed costs for ZK provers often require them to wait until enough transactions are accumulated to fill up a rollup block, introducing additional delays before state roots can be posted to the contract~\cite{medium-tradeoff}.
Moreover, generating a validity proof for rollup blocks often take longer than simply calculating the next state root.
Despite such disadvantages for the execution latency of ZK rollups, future developments in cryptography could bring prover times down, making verification of state roots as fast for the light clients as full nodes in the future.

\subsubsection{Sequencer Selection.}
For user experience, it is crucial to select an honest sequencer that does not censor transactions.
A common feature to all selection algorithms is the use of the parent chain as the medium of agreement.
Such algorithms include holding periodic auctions on the parent chain (\eg, burn auctions~\cite{ethresearch-burn-auction}), using VRFs or random number generators (\eg, RANDAO~\cite{RANDAO}) for rotating sequencers, and selecting the sender of the first transaction in the last rollup block on the parent chain as the sequencer (\cf~\cite{vitalik-rollup,adler2020building,ethresearch-merged-consensus}).
The key requirement for the selection method is Sybil resistance.
For example, in burn auctions, Sybil resistance is ensured by the burning of coins~\cite{eip-1559}.

\subsection{Incentives for the Safety of Optimistic Rollups}
\label{sec:appendix-incentives}

In an optimistic enshrined rollup, security of light clients depends on (i) the existence of an honest watch tower at all times, and (ii) the censorship resistance (liveness) of the parent chain.

One solution to ensure the continuous presence of watch towers is to designate a special set of nodes for this task. 
However, this makes watch towers a central point of failure, vulnerable to bribing and network level attacks.
Rollup applications and liquidity providers might also act as watch towers to avoid losses due to invalid states accepted by the contract.
However, verifier's dilemma shows that, even with rewards offered for detecting invalid state roots or punishments for them, constant verification does not necessarily maximize the utility of stakeholders under any probability of attack by dishonest full nodes due to computational costs~\cite{verifiers-dilemma}.
Thus, to guarantee the presence of watch towers, it is necessary to introduce external incentives that reward verification, or punish inaction, regardless of whether the activity resulted in the detection of an invalid root or not~\cite{medium-incentives}.
For instance, Arbitrum~\cite{arbitrum} and TrueBit~\cite{truebit} increase the watch towers' payoff gap between the verify and non-verify strategies by a factor that does \emph{not} depend on the probability with which invalid roots are posted to the contract.
Arbitrum introduced attention challenges that force verification of roots at all times with the threat of a penalty, whereas Truebit forces full nodes to occasionally post invalid roots to the contract to provide additional incentives to verifiers.
Although each proposal has its own drawbacks and mitigation strategies
(\cf~\cite{medium-incentives} for a comprehensive analysis), they indicate that verifier's dilemma can be mitigated to ensure the presence of honest watch towers.

Incentives on cryptocurrencies is a rich field~\cite{azouvi2020sok}, and it is not clear how the existence of rollups affect the existing incentives for ensuring the protocol's liveness.
For instance, presence of fraud proofs might prompt a proliferation of anti-mechanisms such as bribes payed out to parent chain nodes to censor these proofs~\cite{ethresearch-censor}.
By contrast, nodes enforcing censorship might make themselves more susceptible to DoS attacks~\cite{hacking-censor-resist} and retaliation by honest nodes, implying that in the long run, the nodes might refuse to censor transactions~\cite{moroz2020doublespend}.
Similarly, even in the presence of a $51\%$ attack, validators of PoS Ethereum can invoke \emph{social consensus} to create a `minority soft fork', which financially punishes the censoring majority with slashing conditions such as inactivity leak~\cite{ethresearch-51-attack}.

\subsection{Privacy}
\label{sec:appendix-privacy}

Rollups can help the adoption of privacy preserving computation techniques such as Zexe~\cite{zexe} as experimental platforms.
A current example of a privacy-preserving rollup is Aztec~\cite{aztec}.
With the advance of fully homomorphic encryption, privacy preserving computation techniques can become mainstream among rollups in the future.

\subsection{Deployment}
\label{sec:appendix-deployment}

Capabilities of the execution on rollups affect their adoption by the existing developer community on decentralized systems such as Bitcoin, Ethereum, Polkadot and Cosmos.
For example, existence of EVM compatible execution can facilitate the transfer of existing smart contracts on Ethereum to the rollup, contributing to the community's shift from the parent chain to rollups as the medium of execution.
There are optimistic (\eg Arbitrum~\cite{arbitrum}) rollups providing EVM-compatible execution, and many ZK rollups such as zkSync~\cite{zksync}, Polygon Hermez~\cite{hermez} and Scroll~\cite{scroll} are on track to support smart contracts in Solidity.


\iftoggle{fullversion}{
\paragraph{\bf Acknowledgments.}
We thank Ben Fisch, Mohammad Ali Maddah-Ali and Dionysis Zindros for several insightful discussions on this project.

ENT was supported by the Stanford Center for Blockchain Research.
}{}


\bibliographystyle{ieeetr}
\bibliography{references}





\end{document}